\documentclass[conference,twocolumn]{IEEEtran}

\usepackage{amsmath,amssymb,amsthm}
\usepackage{fontenc,enumerate}
\usepackage{blindtext,color}
\usepackage{etoolbox}
\usepackage{graphicx}
\usepackage{tikz,pgfplots}
\usetikzlibrary{arrows,shapes,chains,matrix,positioning,scopes,patterns,fit}
\usetikzlibrary{decorations.markings,decorations.pathmorphing,backgrounds}
\usetikzlibrary{external}
\usepackage{psfrag,epsfig,epsf,latexsym,hhline,multirow}
\usepgflibrary{shapes}
\pgfplotsset{compat=newest}
\pgfplotsset{plot coordinates/math parser=false}
\newcommand{\mc}[1]{\mathcal{#1}}
\newcommand{\mbf}[1]{\mathbf{#1}}
\newcommand{\coleq}{\mathrel{\mathop:}=}
\newcommand{\Z}{\mathbb{Z}}
\newcommand{\tx}[1]{\text{#1}}

\newcommand{\iid}{i.\@i.\@d.\ }
\newcommand{\ceil}[1]{\lceil{#1}\rceil}

\theoremstyle{definition}\newtheorem{lemma}{Lemma}
\theoremstyle{definition}
\theoremstyle{definition}\newtheorem{theorem}[lemma]{Theorem}
\theoremstyle{definition}
\newtheorem{definition}[lemma]{Definition}
\newtheorem{Example}[lemma]{Example}
\newtheorem{Remark}[lemma]{Remark}

\DeclareMathOperator{\diag}{diag}

\def\ceps{c_{\epsilon}}

\begin{document}
\title{Group Testing using left-and-right-regular sparse-graph codes}
\author{Avinash Vem, Nagaraj T. Janakiraman, Krishna R. Narayanan\\
Department of Electrical and Computer Engineering \\
Texas A\&M University\\
{\tt\small {\{vemavinash,tjnagaraj,krn\}@tamu.edu} }}

\maketitle
\begin{abstract}
We consider the problem of non-adaptive group testing of $N$ items out of which $K$ or less items are known to be defective. We propose a testing scheme based on left-{\em and}-right-regular sparse-graph codes and a simple iterative decoder. We show that for any arbitrarily small $\epsilon>0$ our scheme requires only $m=\ceps K\log \frac{c_1N}{K}$ tests to recover $(1-\epsilon)$ fraction of the defective items with high probability (w.h.p) i.e., with probability approaching $1$ asymptotically in $N$ and $K$, where the value of constants $\ceps$ and $\ell$ are a function of the desired error floor $\epsilon$ and constant $c_1=\frac{\ell}{\ceps}$ (observed to be approximately equal to 1 for various values of $\epsilon$). More importantly the iterative decoding algorithm has a sub-linear computational complexity of $\mc{O}(K\log \frac{N}{K})$ which is known to be optimal. Also for $m=c_2 K\log K\log \frac{N}{K}$ tests our scheme recovers the \textit{whole} set of defective items w.h.p. These results are valid for both noiseless and noisy versions of the problem as long as the number of defective items scale sub-linearly with the total number of items, i.e., $K=o(N)$. The simulation results validate the theoretical results by showing a substantial improvement in the number of tests required when compared to the testing scheme based on left-regular sparse-graphs.
\end{abstract}

\section{Introduction}
The problem of Group Testing (GT) refers to testing a large population of $N$ items for $K$ defective items (or sick people) where grouping multiple items together for a single test is possible. The output of the test is \textit{negative} if all the grouped items are non-defective or else the output is \textit{positive.} In the scenario when $K \ll N$, the objective of GT is to design the testing scheme such that the total number of tests $m$ to be performed is minimized.

This problem was first introduced to the field of statistics by Dorfman \cite{dorfman1943detection} during World War II for testing the soldiers for syphilis without having to test each soldier individually. Since then group testing has found application in wide variety of problems like clone library screening, non-linear optimization, multi-access communication etc.., \cite{du1999combinatorial} and fields like biology\cite{chen2008survey}, machine learning\cite{malioutov2013exact}, data structures\cite{goodrich2005indexing} and signal processing\cite{emad2014poisson}. A comprehensive survey on group testing algorithms, both combinatorial and probabilistic, can be found in \cite{du1999combinatorial,chan2014non,atia2012boolean}. 

In the literature on Group Testing, three kinds of reconstruction guarantees have been considered: combinatorial, probabilistic and approximate. In the combinatorial designs for the GT problem, the probability of recovery for any given defective set should be equal to $1$ whereas in the probabilistic version one is interested in recovering \textit{all} the defective items with high probability (w.h.p) i.e., with probability approaching $1$ asymptotically in $N$ and $K$. Another variant of the probabilistic version is that the probability of recovery is required to be  greater than or equal to $(1-\varepsilon)$ for a given $\varepsilon>0$. For the approximate recovery version one is interested in only recovering a $(1-\epsilon)$ fraction of the defective items (not the whole set of defective items) w.h.p.

For the combinatorial GT the best known lower bound on the number of tests required is $\Omega(K^2\frac{\log N}{\log K})$ \cite{d1982bounds,erdos1985families} whereas the best known achievability bound is $\mc{O}(K^2 \log N)$ \cite{kautz1964nonrandom,porat2011explicit}. Most of these results were based on algorithms relying on exhaustive searches thus have a high computational complexity of atleast $\mc{O}(K^2 N\log N)$. Only recently a scheme with efficient decoding was proposed by Indyk et al., \cite{indyk2010efficiently} where all the defective items are guaranteed to recover using $m=\mc{O}(K^2\log N)$ tests in $\text{poly}(K)\cdot \mc{O}(m \log^2 m )+\mc{O}(m^2)$ time. 

If we consider the probabilistic version of the problem, it was shown in \cite{chan2014non,atia2012boolean} that the number of tests necessary is $\Omega(K\log \frac{N}{K})$ which is the best known lower bound in the literature. And regarding the best known achievability bound Mazumdar \cite{mazumdar2015nonadaptive} proposed a construction that has an asymptotically decaying error probability with $\mc{O}(K\frac{\log^2 N}{\log K})$ tests. For the approximate version it was shown \cite{atia2012boolean} that the required number of tests scale as $\mc{O}(K\log N)$ and to the best of our knowledge this is the tightest bound known.

In \cite{lee2015saffron} authors Lee, Pedarsani and Ramchandran proposed a testing scheme based on \textit{left-regular sparse-graph} codes and a simple iterative decoder based on the\textit{peeling} decoder, which are popular tools in channel coding \cite{richardson2008modern}, for the non-adaptive group testing problem. They refer to the scheme as SAFFRON(\textbf{S}parse-gr\textbf{A}ph codes \textbf{F}ramewrok \textbf{F}or g\textbf{RO}up testi\textbf{N}g), a reference which we will follow through this document. The authors proved that using SAFFRON scheme $m=c_\epsilon K\log N$ number of tests are enough to identify atleast $(1-\epsilon)$ fraction of defective items (the approximate version of GT) w.h.p. The precise value of constant $c_\epsilon$ as a function of the required error floor $\epsilon$ is also given. More importantly the computational complexity of the proposed peeling based decoder is only $\mc{O}(K\log N)$. They also showed that with $m=c\cdot K\log K \log N$ tests  i.e. with an additional $\log K$ factor, the \textit{whole} defective set (the probabilistic version of GT) can be recovered with an asymptotically high probability of $1-\mc{O}(K^{-\alpha})$.

\subsection*{Our Contributions}
In this work, we propose a non-adaptive GT scheme that is similar to the SAFFRON but we employ \textit{left-and-right-regular sparse-graph} codes instead of the left-regular sparse-graph codes and show that we only require $\ceps K\log \frac{N\ell}{\ceps K}$ number of tests for an error floor of $\epsilon$ in the approximate version of the GT problem. Although the testing complexity of our scheme has the same asymptotic order $\mc{O}(K\log N)$ as that of \cite{lee2015saffron}, which as far as we are aware is the best known order result for the required number of tests in the approximate GT, it provides a better explicit upper bound of $\Theta(K\log \frac{N}{K})$ with optimal computational complexity $\mc{O}(K\log \frac{N}{K})$ and also a significant improvement in the required number of tests for finite values of $K,N$. 
Following the approach in \cite{lee2015saffron} we extend our proposed scheme with the singleton-only variant of the decoder to tackle the probabilistic version of the GT problem. In Sec.~\ref{Sec:Singleton-only} we show that for $m=c\cdot K\log K \log \frac{N}{K}$ tests i.e. with an additional $\log K$ factor the \textit{whole} defective set can be recovered w.h.p. Note that the testing complexity of our scheme is only $\log K$ factor away from the best known lower bound of $\Omega(K\log \frac{N}{K})$ \cite{chan2014non} for the probabilistic GT problem. We also extend our scheme to the noisy GT problem, where the test results are corrupted by noise, using an error-correcting code similar to the approach taken in \cite{lee2015saffron}. We demonstrate the improvement in the required number of tests due to \emph{left-}and-\emph{right}-regular graphs  for finite values of $K, N$ via simulations.

\section{Problem Statement}
Formally the group testing problem can be stated as following. Given a total number of $N$ items out of which $K$ are defective, the objective is to perform $m$ different tests and identify the location of the $K$ defective items from the test outputs. For now we consider only the noiseless group testing problem i.e., the result of each test is exactly equal to the boolean OR of all the items participating in the test. 

Let the support vector $\mathbf{x}\in\{0,1\}^{N}$ denote the list of items in which the indices with non-zero values correspond to the defective items. A non-adaptive testing scheme consisting of $m$ tests can be represented by a matrix $\mbf{A}\in\{0,1\}^{m\times N}$ where each row $\mbf{a}_{i}$ corresponds to a test. The non-zero indices in row $\mbf{a}_i$ correspond to the items that participate in $i^{\text{th}}$ test. The output corresponding to vector $\mbf{x}$ and the testing scheme $\mbf{A}$ and can be expressed in matrix form as:
\begin{equation*}
\mbf{y}=\mbf{A\odot x}
\end{equation*}
where $\odot$ is the usual matrix multiplication in which the arithmetic multiplications are replaced by the boolean AND operation and the arithmetic additions are replaced by the boolean OR operation.

%

\section{Review: SAFFRON}
\label{Sec:PriorWork}
As mentioned earlier the SAFFRON scheme \cite{lee2015saffron} is based on left-regular sparse graph codes and is applied for non-adaptive group testing problem. In this section we will briefly review their testing scheme, iterative decoding scheme (reconstruction of $\mbf{x}$ given $\mbf{y}$) and their main results. The SAFFRON testing scheme consists of two stages: the first stage is based on a left-regular sparse graph code which pools the $N$ items into $M$ non-disjoint bins where each item belongs to exactly $\ell$ bins. The second stage comprises of producing $h$ testing outputs at each bin where the $h$ different combinations of the pooled items (from the first stage) at the respective bin are defined according to a universal signature matrix. For the first stage the authors consider a bipartite graph with $N$ variable nodes (corresponding to the $N$ items) and $M$ bin nodes. Each variable node is connected to $\ell$ bin nodes chosen uniformly at random from the $M$ available bin nodes. All the variable nodes (historically depicted on the left side of the graph in coding theory) have a degree $\ell$, hence the left-regular, whereas the degree of a bin node on the right is a random variable in the range $[0:N]$.

\begin{definition}[Left-regular sparse graph ensemble]
Let $\mc{G}_{\ell}(N,M)$ be the ensemble of left-regular bipartite graphs where for each variable node the $\ell$ right node connections are chosen uniformly at random from the $M$ right nodes.
\end{definition}

Let $\mbf{T}_{G}\in\{0,1\}^{M\times N }$ be the adjacency matrix corresponding to a graph $G\in\mc{G}_{\ell}(N,M)$ i.e., each column in $\mbf{T}_{G}$ corresponds to a variable node and has exactly $\ell$ ones. Let the rows in matrix $\mbf{T}_{G}$ be given by $\mbf{T}_{G}=[\mbf{t}^{T}_1,\mbf{t}^{T}_{2},\dots, \mbf{t}^{T}_{M}]^{T}$. For the second stage let the universal signature matrix defining the $h$ tests at each bin be $\mbf{U}\in\{0,1\}^{h\times N}$. Then the overall testing matrix $\mbf{A}\coleq [\mbf{A}_{1}^{T},\ldots,\mbf{A}_{M}^{T}]^T$ where $\mbf{A}_{i}=\mbf{U} \diag (\mbf{t}_i)$ of size $h\times N$ defines the $h$ tests at $i^{\text{th}}$ bin. Thus the total number of tests is $m=M\times h$. 
 
 The signature matrix 	$\mbf{U}$ in a more general setting with parameters $r$ and $p$ can be given by
 \begin{align}
\label{Eqn:SignatureMatrix}
\mbf{U}_{r,p}=\begin{bmatrix}
\mbf{b}_1  & \mbf{b}_2 &\cdots & \mbf{b}_r \\
\overline{\mbf{b}}_1 & \overline{\mbf{b}}_2 & \cdots & \overline{\mbf{b}}_r\\
\mbf{b}_{\pi^{1}_{1}} & \mbf{b}_{\pi^{1}_{2}} & \cdots & \mbf{b}_{\pi^{1}_{r}}\\
\overline{\mbf{b}}_{\pi^{1}_{1}} & \overline{\mbf{b}}_{\pi^{1}_{2}} & \cdots & \overline{\mbf{b}}_{\pi^{1}_{r}}\\
\cdots &  &\vdots \\
\mbf{b}_{\pi^{p-1}_{1}} & \mbf{b}_{\pi^{p-1}_{2}} & \cdots & \mbf{b}_{\pi^{p-1}_{r}}\\
\overline{\mbf{b}}_{\pi^{p-1}_{1}} & \overline{\mbf{b}}_{\pi^{p-1}_{2}} & \cdots & \overline{\mbf{b}}_{\pi^{p-1}_{r}}
\end{bmatrix}
\end{align}  
where $\mbf{b}_{i}\in\{0,1\}^{\ceil{\log_{2}r}}$ is the binary expansion vector for $i$ and $\overline{\mbf{b}}_{i}$ is the complement of $\mbf{b}_{i}$. $\mbf{\pi}^{k}=[\pi^{k}_1,\pi^{k}_2,\ldots,\pi^{k}_r]$ denotes a permutation chosen at random from symmetric group $S_{r}$. Henceforth $\mbf{U}_{r,p}$ will refer to either the ensemble of matrices generated over the choices of the permutations $\pi^{k}$ for $k\in[1:p-1]$ or a matrix picked uniformly at the random from the said ensemble. The reference should be sufficiently clear from the context. In the SAFFRON scheme the authors employed a signature matrix from $\mbf{U}_{r,p}$ with $r=N$ and $p=3$ thus resulting in a $\mbf{U}$ of size $h \times N$ with $h=6\log_{2}N$. 

\subsection*{Decoding}
Before describing the decoding process let us review some terminology. A bin is referred to as a \textit{singleton} if there is exactly one non-zero variable node connected to the bin and similarly referred to as a \textit{double-ton} in case of two non-zero variable nodes. In the case where we know the identity of one of them leaving the decoder to decode the identity of the other one, the bin is referred to as a \textit{resolvable double-ton}. And if the bin has more than two non-zero variable nodes attached we refer to it as a \textit{multi-ton}. First part of the decoder which is referred to as bin decoder will be able to detect and decode exactly the identity of the non-zero variable nodes connected to the bin if and only if the bin is a singleton or a resolvable double-ton. If the bin is a multi-ton the bin decoder will detect it neither as a singleton nor a resolvable double-ton with high probability. The second part of the decoder which is commonly referred to as peeling decoder \cite{li2015subisit}, when given the identities of some of the non-zero variable nodes by the bin decoder, identifies the bins connected to the recovered variable nodes and looks for newly uncovered resolvable double-ton in these bins. This process of recovering new non-zero variable nodes from already discovered non-zero variable nodes proceeds in an iterative manner (referred to as peeling off from the graph historically). For details of the decoder we refer the reader to \cite{lee2015saffron}.

The overall group testing decoder comprises of these two decoders working in conjunction as follows. In the first and foremost step, given the $m$ tests output, the bin decoder is applied on the $M$ bins and the set of variable nodes that are connected to singletons are decoded and output. We denote the decoded set of non-zero variable nodes as $\mc{D}$. Now in an iterative manner, at each iteration, a variable node from $\mc{D}$ is considered and the bin decoder is applied on the bins connected to this variable node.
The main idea is that if one of these bins is detected as a resolvable double-ton thus resulting in decoding a new non-zero variable node. The considered variable node in the previous iteration is moved from $\mc{D}$ to a set of peeled off variable nodes $\mc{P}$ and the newly decoded non-zero variable node in the previous iteration, if any, will be placed in set $\mc{D}$ and continue to the next iteration. The decoder is terminated when $\mc{D}$ is empty and is declared successful if the set $\mc{P}$ equals the set of defective items. 

\begin{Remark}
 Note that we are not literally peeling off the decoded nodes from the graph because of the \textit{non-linear} OR operation on the non-zero variable nodes at each bin thus preventing us in subtracting the effect of the non-zero node from the measurements of the bin node unlike in the problems of compressed sensing or LDPC codes on binary erasure channel.
\end{Remark}

Now we state the series of lemmas and theorems from \cite{lee2015saffron} that enabled the authors to show that their SAFFRON scheme with the described peeling decoder solves the group testing problem with $c\cdot K\log N$ tests and $\mc{O}(K\log N)$ computational complexity.

\begin{lemma}[Bin decoder analysis]
\label{Lem:BinDecoderAnalysis}
For a signature matrix $\mbf{U}_{r,p}$ as described in \eqref{Eqn:SignatureMatrix}, the bin decoder successfully detects and resolves if the bin is either a singleton or a resolvable double-ton. In the case of the bin being a multi-ton, the bin decoder declares a wrong hypothesis of either a singleton or a resolvable double-ton with a probability no greater than $\frac{1}{r^{p-1}}$.
\end{lemma}
\begin{proof}
This result was proved in \cite{lee2015saffron} for the choice of parameters $r=N$ and $p=3$. The extension of the result to general $r,p$ parameters is straight forward.
\end{proof}

For convenience the performance of the peeling decoder is analyzed independently of the bin decoder i.e., a peeling decoder is considered which assumes that the bin decoder is working accurately which will be referred to as \textit{oracle based peeling decoder}. Another simplification is that a pruned graph is considered where all the zero variable nodes and their respective edges are removed from the graph. Also the oracle based peeling decoder is assumed to decode a variable node if it is connected to a bin node with degree one or degree two with one of them already decoded, in an iterative fashion. Any right node with more than degree two is untouched by this oracle based peeling decoder. It is easy to verify that the original decoder with accurate bin decoding is equivalent to this simplified oracle based peeling decoder on a pruned graph.
\begin{definition}[Pruned graph ensemble]
Let the pruned graph ensemble $\tilde{\mc{G}}_l(N,K,M)$ be the set of all bipartite graphs obtained from removing a random $N-K$ subset of variable nodes from a graph from the ensemble $\mc{G}_{\ell}(N,M)$. Note that graphs from the pruned ensemble have $K$ variable nodes. 
\end{definition}	

Before we analyze the pruned graph ensemble let us define the right-node degree distribution (d.d) of an ensemble as $R(x)=\sum_{i}R_i x^i$ where $R_i$ is the probability that a right-node in any graph from the ensemble has degree $i$. Similarly the edge d.d $\rho(x)=\sum_{i}\rho_ix^{i-1}$ is defined where $\rho_i$ is the probability that a random edge in the graph is connected to a right-node of degree $i$. Note that the left-degree distribution is regular (i.e. $L(x)=x^\ell$) even for the pruned graph ensemble and hence is not specifically discussed.

\begin{lemma}[Edge d.d of Pruned graph]
\label{Lem:EdgeddSAFFRON}
For the pruned ensemble $\tilde{\mc{G}}_{\ell}(N,K,M)$, it was shown that in the limit $K,N\rightarrow\infty$, $\rho_{1}=e^{-\lambda}$ and $\rho_{2}=\lambda e^{-\lambda}$ where $\lambda=\ell/\ceps$ for $M=\ceps K$ for any constant $\ceps$. 
\end{lemma}

\begin{lemma}
\label{Lem:PeelingAnalysisLeftRegular}
For the pruned graph ensemble $\tilde{\mc{G}}_{\ell}(N, K,M)$ the oracle-based peeling decoder fails to peel off atleast $(1-\epsilon)$ fraction of the variable nodes with exponentially decaying probability if $M\geq \ceps K$ where the required $\ceps$ and $\ell$ for various values of $\epsilon$ are given in Table. \ref{Table:constantsDE}.
\end{lemma}
\begin{proof}
Instead of reworking the whole proof here from \cite{lee2015saffron}, we will list the main steps involved in the proof which we will use further along. Let $p_j$ be the probability that a random defective item is not identified at iteration $j$ of the decoder, in the limit $N \text{ and } K\rightarrow \infty$. Then one can write the density evolution (DE) equations relating $p_{j+1}$ to $p_{j}$ as 
\begin{align*}
p_{j+1}=\left[1-(\rho_1+\rho_2(1-p_j))\right]^{\ell-1}.
\end{align*}
For this DE, we can see that $0$ is not a fixed point and hence $p_j\nrightarrow 0$ as $j\rightarrow\infty$. Therefore numerically optimizing the values of $\ceps$ and $\ell$ such that $\lim_{j\rightarrow\infty}p_j\leq \epsilon$ gives the optimal values for $\ceps$ and $\ell$ given in Table. \ref{Table:constantsDE}. It was also shown \cite{lee2015saffron,richardson2008modern} that for such sparse graph systems the actual fraction of the undecoded variable nodes deviates from the average undecoded fraction of the variable nodes given by the DE with exponentially low probability. 
\end{proof}

\begin{table}[t]
\centering
\begin{tabular}{| c | c | c | c | c | c | c | c | }
\hline
$\epsilon$ & $10^{-3}$ & $10^{-4}$ & $10^{-5}$ & $10^{-6}$ &$ 10^{-7}$ & $10^{-8}$ & $10^{-9}$ \\ \hline
$c_1(\epsilon)$ & 6.13 & 7.88 & 9.63 & 11.36 & 13.10 & 14.84 & 16.57 \\ \hline
 $\ell$ & 7 & 9 & 10 & 12 & 14 & 15 & 17 \\ \hline
\end{tabular}
\vspace{1ex}
\caption{Constants for various error floor values}
\label{Table:constantsDE}
\end{table}

Combining the lemmas and remarks above, the main result from \cite{lee2015saffron} can be summarized as below.
\begin{theorem}
A random testing matrix from the SAFFRON scheme with $m=6\ceps K \log_{2}N$ tests recovers atleast $(1-\epsilon)$ fraction of the defective items w.h.p of atleast $1-O(\frac{K}{N^2})$. The computational complexity of the decoding scheme is $O(K\log N)$. The constant $\ceps$ is given in Table. \ref{Table:constantsDE} for some values of $\epsilon$.
\end{theorem}

\section{Proposed Scheme}
The main difference between the SAFFRON scheme described in Sec.~\ref{Sec:PriorWork} and our proposed scheme is that we use a left-and-right-regular sparse-graph instead of left-regular sparse-graph in the first stage for the binning operation.

\begin{definition}[Left-and-right-regular sparse graph ensemble]
We define $\mc{G}_{\ell,r}(N,M)$ to be the ensemble of left-and-right-regular graphs where the $N\ell$ edge connections from the left and $Mr(=N\ell)$ edge connections from the right are paired up according to a permutation $\pi$ chosen at random from $S_{N\ell}$. 
\end{definition}

 Let $\mbf{T}_{G}\in\{0,1\}^{M\times N }$ be the adjacency matrix corresponding to a graph $G\in\mc{G}_{\ell,r}(N,M)$ i.e., each column in $\mbf{T}_{G}$ corresponding to a variable node has exactly $\ell$ ones and each row corresponding to a bin node has exactly $r$ ones. And let the universal signature matrix be $\mbf{U}\in\{0,1\}^{h \times r}$ chosen from the $\mbf{U}_{r,p}$ ensemble. Then the overall testing matrix $\mbf{A}\coleq [\mbf{A}_{1}^{T},\ldots,\mbf{A}_{M}^{T}]^T$ where $\mbf{A}_{i}\in\{0,1\}^{h\times N}$ defining the $h$ tests at $i^{\text{th}}$ bin is given by
 \begin{align}
 \mbf{A}_i&=[\mbf{0},\ldots,\mbf{0},\mbf{u}_1, \mbf{0},\ldots, \mbf{u}_2,\mbf{0}, \ldots, \mbf{u}_{r}],\quad \text{where}\label{Eqn:TestingMatrixDefn}\\
\mbf{t}_i &= [0,\ldots,0,\hspace{0.6ex}1,\hspace{0.9ex} 0, \ldots,\hspace{0.6ex}1,\hspace{0.9ex}0, \ldots, \hspace{0.9ex}1].\nonumber
 \end{align}
Note that $\mbf{A}_i$ is defined by placing the $r$ columns of $\mbf{U}$ at the $r$ non-zero indices of $\mbf{t}_i$ and the remaining are padded with zero columns. We can observe that the total number of tests for this scheme is $m=M\times h$ where $h=2p\log_2 r$.

\begin{Example}\label{exmp:tensor}
Let us look at an example for $(N,M)=(6,3)$ and $(\ell,r)=(2,4)$. Then the adjacency matrix $\mbf{T}_G$ of a graph $G\in\mathcal{G}_{2,4}(6,3)$ and a signature matrix $\mbf{U}\in \{0,1\}^{4\times 3}$ for $p=1$ and $\log_2 r = 2$ are given by
 \[
 \mbf{T}_G = \begin{bmatrix}
1 & 1 & 0 & 1 & 0 & 1    \\
0 & 1 & 1 & 1& 1 & 0 \\
1 & 0 & 1 & 0 & 1 & 1
\end{bmatrix}
\mbf{U} = \begin{bmatrix}
0 & 0 & 1  & 1\\
0 & 1 & 0  & 1\\
1 & 1 & 0  & 0\\
1 & 0 & 1  & 0
\end{bmatrix}.
 \] 
Then, the measurement matrix $\bf A	$ with $m = 2pM\ceil{\log_2 r}=12$ tests is given by
\[ \mbf{A} = \begin{bmatrix}
0 & 0 & 0 & 1 & 0 & 1\\
0 & 1 & 0 & 0 & 0 & 1\\
1 & 1 & 0 & 0 & 0 & 0\\
1 & 0 & 0 & 1 & 0 & 0\\
0 & 0 & 0 & 1& 1 & 0 \\
0 & 0 & 1 & 0& 1 & 0 \\
0 & 1 & 1 & 0& 0 & 0 \\
0 & 1 & 0 & 1& 0 & 0 \\
0 & 0 & 0 & 0 & 1 & 1\\
0 & 0 & 1 & 0 & 0 & 1\\
1 & 0 & 1 & 0 & 0 & 0\\
1 & 0 & 0 & 0 & 1 & 0\\
\end{bmatrix}
 \]
 \end{Example}

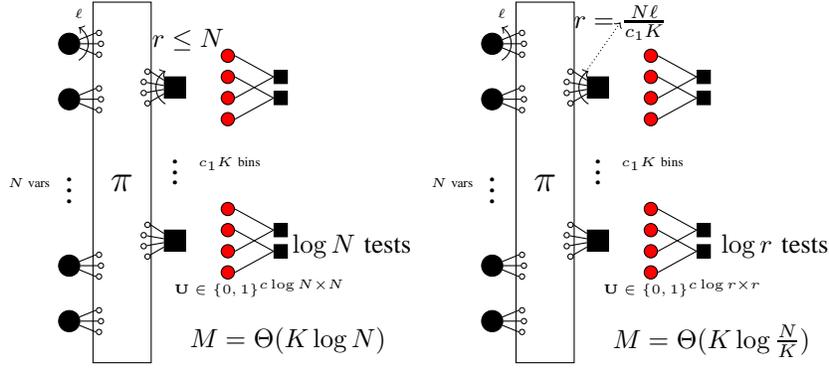
\begin{figure*}[t!]
\centering \scalebox{1}{\begin{tikzpicture}
\def\emphsize{\normalsize}
\def\horzgap{40pt}; 
\def \gapVN{21pt}; 
\def \gapCN{	29pt}; 

\def \textoffs{9pt}; 
\def\nodewidth{8pt};
\def\nodewidthsm{5pt}
\def\edgewidth{2pt};
\def\ext{12pt};
\def\moveX {20pt};
\def\moveYa{4pt}
\def\moveYb{12pt}
\def\arcradius{0.6*\ext}

\def\xmovement{4*\horzgap}

\def \n {8};
\def\ldeg{3};
\def \m {4};
\def\rdeg{6};
\def\langle{40};
\def\langle{20};

\tikzstyle{check} = [rectangle, draw,  inner sep=0mm, fill=black,minimum height=\nodewidth,minimum width=\nodewidth]
\tikzstyle{checksm} = [rectangle, draw, inner sep=0mm, fill=black,minimum height=\nodewidthsm, minimum width=\nodewidthsm]

\tikzstyle{bit} = [circle, draw, inner sep=0mm, fill=black, minimum size=\nodewidth]
\tikzstyle{bitsm} = [circle, draw, inner sep=0mm,fill=red, minimum size=\nodewidthsm]
\tikzstyle{edgesock} = [circle, inner sep=0mm, minimum size=\edgewidth,draw, fill=white]

\foreach \vn in {2,3,6,7}{
 \node[bit] (vn\vn) at (0,\vn*\gapVN) {};
\path (vn\vn) ++(20:\ext)node (evA\vn) [edgesock] {};
\path (vn\vn) ++(0:\ext) node (evB\vn) [edgesock] {};
\path (vn\vn) ++(-20:\ext) node (evC\vn) [edgesock] {};

  \draw (vn\vn) -- (evA\vn.west); 
  \draw (vn\vn) -- (evB\vn.west); 
  \draw (vn\vn) -- (evC\vn.west); 
}
\draw[->,thin] (vn7)++(-60:\arcradius) arc (-60:60:\arcradius)node[above]{\tiny{$\ell$}};

\path (vn3)--node(vndots) {\Large{$\vdots$}} (vn6);
\node[left =-0.2*\nodewidth of vndots](){\tiny{$N$ vars}};

\foreach \cn in {2,4}{
\node[check] (cn\cn) at (\horzgap,0.2in+\cn*\gapCN) {};

\path (cn\cn) ++(150:\ext) node (ecA\cn) [edgesock] {};
\path (cn\cn) ++(170:\ext) node (ecB\cn) [edgesock] {};
\path (cn\cn) ++(190:\ext) node (ecC\cn) [edgesock] {};
\path (cn\cn) ++(210:\ext) node (ecD\cn) [edgesock] {};

  \draw (cn\cn) -- (ecA\cn.east); 
  \draw (cn\cn) -- (ecB\cn);  
  \draw (cn\cn) -- (ecC\cn); 
    \draw (cn\cn) -- (ecD\cn); 
}

\draw[->,thin] (cn4)++(240:\arcradius) arc (240:120:\arcradius);
\path(cn4)++(5pt,2*\textoffs)node(){\emphsize{$r\leq N$}};
\path (cn2)--node(cndots) {\Large{$\vdots$}} (cn4);
\node [right=0.01*\nodewidth of cndots]{\tiny{$c_1K$ bins}};

\node[draw,minimum width=\horzgap-1.5*\ext,minimum height=6.5*\gapVN](perm) at (0.5*\horzgap,4.5*\gapVN){\Large{$\pi$}};

\foreach \cn in {2,4}{
\path (cn\cn) ++(\moveX,\moveYb) node (bitnA\cn) [bitsm] {};
\path (cn\cn) ++(\moveX,\moveYa) node (bitnB\cn) [bitsm] {};
\path (cn\cn) ++(\moveX,-\moveYa) node (bitnC\cn) [bitsm] {};
\path (cn\cn) ++(\moveX,-\moveYb) node (bitnD\cn) [bitsm] {};

\path (bitnB\cn) ++(\moveX,0) node (checknB\cn) [checksm] {};
\path (bitnC\cn) ++(\moveX,0) node (checknC\cn) [checksm] {};

\draw (bitnA\cn.east)--(checknB\cn.west);
\draw (bitnC\cn.east)--(checknB\cn.west);
\draw (bitnB\cn.east)--(checknC\cn.west);
\draw (bitnD\cn.east)--(checknC\cn.west);
}
\path(bitnD2)++(0.6*\moveX,-0.3*\moveX) node(){\tiny{$\mathbf{U}\in\{0,1\}^{c\log N \times N }$}};

\path (checknC2)++(3*\textoffs,0.5*\moveYa)node{\emphsize{$\log N$ tests}};
\node [below left=3*\nodewidth and 2*\nodewidth of bitnD2,anchor=west]{\emphsize{$M=\Theta(K\log N)$} };

\foreach \vn in {2,3,6,7}{
 \node[bit] (Bvn\vn) at (\xmovement,\vn*\gapVN) {};
\path (Bvn\vn) ++(20:\ext)node (BevA\vn) [edgesock] {};
\path (Bvn\vn) ++(0:\ext) node (BevB\vn) [edgesock] {};
\path (Bvn\vn) ++(-20:\ext) node (BevC\vn) [edgesock] {};

  \draw (Bvn\vn) -- (BevA\vn.west); 
  \draw (Bvn\vn) -- (BevB\vn.west); 
  \draw (Bvn\vn) -- (BevC\vn.west); 
}
\draw[->,thin] (Bvn7)++(-60:\arcradius) arc (-60:60:\arcradius)node[above]{\tiny{$\ell$}};

\path (Bvn3)--node(Bvndots) {\Large{$\vdots$}} (Bvn6);
\node[left =-0.2*\nodewidth of Bvndots](){\tiny{$N$ vars}};

\foreach \cn in {2,4}{
\node[check] (Bcn\cn) at (\xmovement+\horzgap,0.2in+\cn*\gapCN) {};

\path (Bcn\cn) ++(150:\ext) node (BecA\cn) [edgesock] {};
\path (Bcn\cn) ++(170:\ext) node (BecB\cn) [edgesock] {};
\path (Bcn\cn) ++(190:\ext) node (BecC\cn) [edgesock] {};
\path (Bcn\cn) ++(210:\ext) node (BecD\cn) [edgesock] {};

  \draw (Bcn\cn) -- (BecA\cn.east); 
  \draw (Bcn\cn) -- (BecB\cn);  
  \draw (Bcn\cn) -- (BecC\cn); 
  \draw (Bcn\cn) -- (BecD\cn); 
}

\draw[->,thin] (Bcn4)++(240:\arcradius) arc (240:120:\arcradius);
\draw[->,densely dotted](Bcn4)++(140:\arcradius)--++(14pt,20pt)node[](){\emphsize{$r=\frac{N\ell}{c_1K}$}};
\path (Bcn2)--node(Bcndots) {\Large{$\vdots$}} (Bcn4);
\node [right=0.01*\nodewidth of Bcndots]{\tiny{$c_1K$ bins}};

\node[draw,minimum width=\horzgap-1.5*\ext,minimum height=6.5*\gapVN](Bperm) at (\xmovement+0.5*\horzgap,4.5*\gapVN){\Large{$\pi$}};

\foreach \cn in {2,4}{
\path (Bcn\cn) ++(\moveX,\moveYb) node (BbitnA\cn) [bitsm] {};
\path (Bcn\cn) ++(\moveX,\moveYa) node (BbitnB\cn) [bitsm] {};
\path (Bcn\cn) ++(\moveX,-\moveYa) node (BbitnC\cn) [bitsm] {};
\path (Bcn\cn) ++(\moveX,-\moveYb) node (BbitnD\cn) [bitsm] {};

\path (BbitnB\cn) ++(\moveX,0) node (BchecknB\cn) [checksm] {};
\path (BbitnC\cn) ++(\moveX,0) node (BchecknC\cn) [checksm] {};

\draw (BbitnA\cn.east)--(BchecknB\cn.west);
\draw (BbitnC\cn.east)--(BchecknB\cn.west);
\draw (BbitnB\cn.east)--(BchecknC\cn.west);
\draw (BbitnD\cn.east)--(BchecknC\cn.west);
}
\path(BbitnD2)++(0.6*\moveX,-0.3*\moveX) node(){\tiny{$\mathbf{U}\in\{0,1\}^{c\log r\times r}$}};

\path (BchecknC2)++(3*\textoffs,0.5*\moveYa)node{\emphsize{$\log r$ tests}};
\node [below left=3*\nodewidth and 2*\nodewidth of BbitnD2,anchor=west]{\emphsize{$M=\Theta(K \log \frac{N}{K})$} };

\end{tikzpicture}}
\caption{Illustration of the main differences between SAFFRON \cite{lee2015saffron} on the left and our regular-SAFFRON scheme on the right. In both the schemes the peeling decoder on sparse graph requires $\Theta(K)$ bins. But for the bin decoder part, in SAFFRON scheme the right degree is a random variable with a maximum value of $N$ and thus requires $\Theta(\log N)$ tests at each bin. Whereas our scheme based on right-regular sparse graph has a constant right degree of $\Theta(\frac{N}{K})$ and thus requires only $\Theta(\log \frac{N}{K})$ tests at each bin. Thus we can improve the number of tests from $\Theta (K\log N)$ to $\Theta(K\log \frac{N}{K})$.}
\end{figure*}

\begin{definition}[Regular-SAFFRON]
\label{Def:RegSaffron}
Let the ensemble of testing matrices be $\mc{G}_{\ell,r}(N,M)\times \mbf{U}_{r,p}$ where a graph $G$ from $\mc{G}_{\ell,r}(N,M)$ and a signature matrix $\mbf{U}$ from $\mbf{U}_{r,p}$ are chosen at random and the testing matrix $\mbf{A}$ is defined according to Eq. \eqref{Eqn:TestingMatrixDefn}. Note that the total number of tests is $2pM\log_2 r$ where $r=\frac{N\ell}{M}$.
\end{definition}

For the regular-SAFFRON testing ensemble defined in Def. \ref{Def:RegSaffron}, we employ the iterative decoder described in Sec.~\ref{Sec:PriorWork}. Similar to the SAFFRON scheme we will analyze the peeling decoder and the bin decoder separately and union bound the total error probability of the decoding scheme. As we have already mentioned the analysis of just the peeling decoder part can be carried out by considering a \textit{simplified oracle-based peeling decoder} on a pruned graph with only the non-zero variable nodes remaining. 

\begin{definition}[Pruned graph ensemble]
We will define the pruned graph ensemble $\tilde{\mc{G}}_{\ell,r}(N,K,M)$ as the set of all graphs obtained from removing a random $N-K$ subset of variable nodes from a graph from left-and-right-regular sparse-graph ensemble $\mc{G}_{\ell,r}(N,M)$.
\end{definition}

Note that graphs from the pruned ensemble have $K$ variable nodes with a degree $\ell$ whereas the right degree is not regular anymore. 

\begin{lemma}[Edge d.d of pruned graph]
\label{Lem:EdgeDDPrunedGraph}
For the pruned graph ensemble $\tilde{\mc{G}}_{\ell,r}(N,K,M)$ it can be shown in the limit $K,N\rightarrow\infty$ and $K=o(N)$ that the edge d.d coefficients approach $\rho_{1}=e^{-\lambda}$ and $\rho_{2}=\lambda e^{-\lambda}$ where $\lambda=\ell/c$ for the choice of $M=cK$, $c$ being some constant.
\end{lemma}
\begin{proof}
We will first derive $R(x)$ for the pruned graph ensemble and then use the relation $\rho(x)=\frac{R'(x)}{R'(1)}$\cite{richardson2008modern}  to derive the edge d.d. Note that all the bin nodes have a uniform degree $r$ before pruning. In the pruning operation we are removing a $N-K$ subset of variable nodes at random which means from the bin node perspective, in an asymptotic sense, this is equivalent to removing each connected edge with a probability $1-\beta$ where $\beta\coleq \frac{K}{N}$. Under this process the right-node d.d can be written as
\begin{align}
R_1&=r\beta(1-\beta)^{r-1},\quad \text{ and similarly}\label{Eqn:Deg1ChkDistribution}\\
R_i &=\binom{r}{i} \beta^{i}(1-\beta)^{r-i} ~~\forall i<=r \nonumber
\end{align}
thus giving us $R(x)=(\beta x+(1-\beta))^{r}$. This gives us 
\begin{align*}
\rho(x)&=\frac{r\beta(\beta x+(1-\beta))^{r-1}}{r\beta}\\
          &=(\beta x+(1-\beta))^{r-1}.
\end{align*}
Thus we can compute that $\rho_1=(1-\beta)^{r-1}$ and $\rho_2=(r-1)\beta(1-\beta)^{r-2}$. For $M=c K$ we evaluate these quantities in the limit $K,N\rightarrow \infty$ as
\begin{align*}
\lim_{K,N\rightarrow \infty} \rho_1&=\lim_{K,N\rightarrow \infty} \left(1-\frac{K}{N}\right)^{\frac{N\ell}{c K}-1}\\
&=e^{-\lambda} \qquad \text{ where } \lambda=\frac{\ell}{c}
\end{align*}
Similarly we can show $\lim_{K,N\rightarrow \infty}\rho_2=\lambda e^{-\lambda}$.
\end{proof}

Note that even if our initial ensemble is left-and-right-regular the pruned graph ensemble has asymptotically the same degree distribution as in the SAFFRON scheme where the initial ensemble is left-regular.

\begin{lemma}
\label{Lem:PeelingRegularAnalysis}
For the pruned graph ensemble $\tilde{\mc{G}}_{\ell,r}(N, K,M)$ the oracle-based peeling decoder fails to peel off atleast $(1-\epsilon)$ fraction of the variable nodes with exponentially decaying probability for $M=\ceps K$ where $\ell, \ceps$ for various $\epsilon$ is given in Table. \ref{Table:constantsDE}.
\end{lemma}
\begin{proof}
We showed in Lemma. \ref{Lem:EdgeDDPrunedGraph} that, in the limit of $K,N\rightarrow\infty$, the edge degree distribution coefficients $\rho_1$ and $\rho_2$ approach the same values as in the SAFFRON scheme (see Lem. \ref{Lem:EdgeddSAFFRON}). Now we follow the exact same approach as that of Lem. \ref{Lem:PeelingAnalysisLeftRegular} where the limiting values of $\rho_1=e^{-\lambda}$ and $\rho_2=\lambda e^{-\lambda}$ are used in the DE equations to show that for the given values of $\ell$ and $\ceps$ $\lim_{j\rightarrow\infty}p_j \leq \epsilon$. 
\end{proof}

\begin{theorem}
\label{Thm:NoiselessMain}
Let $p\in\Z$ such that $K=o(N^{1-1/p})$. A random testing matrix from the proposed regular SAFFRON ensemble $\mc{G}_{\ell,\frac{N\ell}{\ceps K}}(N,\ceps K)\times \mbf{U}_{\frac{N\ell}{\ceps K},p}$ with $m=c\cdot K\log_{2}\frac{c_2 N}{K}$ tests recovers atleast $(1-\epsilon)$ fraction of the defective items w.h.p. The computational complexity of the decoding scheme is $\mc{O}(K\log \frac{N}{K})$. The constants are $c=2p\ceps, c_2=\frac{\ell}{\ceps}$ where $\ell$ and $\ceps$ for various values of $\epsilon$ are given in Table. \ref{Table:constantsDE}.
\end{theorem}
\begin{proof}
It remains to be shown that for the proposed regular SAFFRON scheme the total probability of error vanishes asymptotically in $K$ and $N$. Let $E_1$ be the event of oracle-based peeling decoder terminating without recovering atleast $(1-\epsilon)K$ variable nodes. Let $E_2$ be the event of the bin decoder making an error during the entirety of the peeling process and $E_{\tx{bin}}$ be the event of one instance of bin decoder making an error. The total probability of error $P_e$ can be upper bounded by
\begin{align*}
P_e &\leq \tx{Pr}(E_1)+ \tx{Pr}(E_2)\\
               &\leq \tx{Pr}(E_1)+ K\ell\tx{ Pr}(E_{\tx{bin}})\\
               &\in O\left(\frac{K^{p}}{N^{p-1}}\right)
\end{align*}
where the second inequality is due to the union bound over a maximum of $K\ell$ (number of edges in the pruned graph) instances of bin decoding. The third line is due to the fact that $\tx{Pr}(E_1)$ is exponentially decaying in $K$ (see Lemma. \ref{Lem:PeelingRegularAnalysis}) and $\tx{Pr}(E_{\tx{bin}})=(\frac{\ceps K}{N\ell})^{p-1}$ (see Lemma. \ref{Lem:BinDecoderAnalysis} and Def. \ref{Def:RegSaffron})
\end{proof}

\section{Total recovery: Singleton-Only Variant}
\label{Sec:Singleton-only}
 In this section we will look at the proposed regular-SAFFRON scheme but with a decoder that uses only the singleton bins. To elaborate, the only difference is in the decoder which is not iterative in this framework and recovers the variable nodes connected to only the singleton bin nodes and terminates. We will refer to this scheme as \textit{singleton-only} regular-SAFFRON scheme. The trade-off is that we can now recover the \textit{whole} defective set instead of just a large fraction of the defective items with an additional $\log K$ factor tests. Since we do not need to be able to recover resolvable double-tons we only need $2\log_2 r$ number of tests at each bin i.e. we choose $p=1$ for the signature matrix in Eqn. \eqref{Eqn:SignatureMatrix}.
\begin{theorem}
Let $K=o(N)$. For $M=c_\alpha K \log K$ and $(\ell,r)=(c_\alpha \log K,\frac{N}{K})$ a random testing matrix from the regular SAFFRON ensemble $\mc{G}_{\ell,r}(N,M)\times \mbf{U}_{r,1}$ with $m=2c_\alpha K\log K \log_2 \frac{N}{K}$ tests the singleton-only decoder fails to recover all the non-zero variable nodes with a vanishing probability of $\mc{O}(K^{-\alpha})$ where $c_\alpha=e(1+\alpha)$.
\end{theorem}
\begin{proof}
First we observe that for the choice of $(\ell,r)=(c_\alpha \log K,\frac{N}{K})$ number of bins $M=\frac{N\ell}{r}=c_\alpha K \log K$ and the number of tests in each bin is $2\log_2 \frac{N}{K}$. From Lem. \ref{Lem:BinDecoderAnalysis} we know that a singleton bin is guaranteed to be decoded by the bin decoder. Thus it is enough if we show that for this choice for the number of bins $M$ all the variable nodes in the pruned graph are connected to atleast one singleton bin w.h.p of $1-\mc{O}(K^{-\alpha})$.

In the pruned graph ensemble, for any particular variable node, the probability that any of the $\ell$ connected bit nodes are not a singleton can be given by $(1-R_1)^\ell$ where $R_1$ is the probability that a bin node in the pruned graph ensemble is a singleton. In the limit $K,N\rightarrow \infty$ the value of $R_1$ approaches (from Eq. \ref{Eqn:Deg1ChkDistribution})
\begin{align*}
R_1&=\lim _{K,N\rightarrow\infty}r\beta(1-\beta)^{r-1}\\
     &=\lim _{\frac{N}{K}\rightarrow\infty}\left(1-\frac{K}{N}\right)^{\frac{N}{K}-1}\\
     &= e^{-1}
\end{align*} 
By using union bound over all the $K$ variable nodes in the pruned graph, the probability $P_e$ that the singleton-only decoder fails to recover a defective item can be bounded by
\begin{align*}
P_e&\leq K(1-R_1)^\ell \\
&=\mc{O}\left(K(1-e^{-1})^{e(1+\alpha)\log K}\right)\\
&=\mc{O}\left(Ke^{-e^{-1}e(1+\alpha)\log K}\right)\\
&=\mc{O}\left(K^{-\alpha}\right).
\end{align*}
In third line we used $(1-x)\leq e^{-x}$.
\end{proof}

\section{Robust Group Testing}
\label{Sec:NoisyGroupTesting}
In this section we extend our scheme to the group testing problem where the test results can be noisy. Formally, the signal model can be described as 
\begin{align*}
\mbf{y=A}\odot \mbf{x+w},
\end{align*}
where $\mbf{w}\in\{0,1\}^N$ is an \iid noise vector distributed according to Bernoulli distribution with parameter $0<q<\frac{1}{2}$ and the addition is over binary field.

\subsection*{Testing Scheme}
In \cite{lee2015saffron} for the robust group testing problem, the signature matrix used for noiseless group testing problem is modified using an error control code such that it can handle singletons and resolvable doubletons in the presence of noise. The binning operation as defined by the bipartite graph is exactly identical to that of noiseless case. We describe the modifications to the signature matrix and the bin detection decoding scheme as given in \cite{lee2015saffron} for the sake of completeness and then state the performance bounds for our scheme for the noisy group testing problem.

Let $\mc{C}_n$ be a binary error-correcting code with the following definition:
\begin{itemize}
\item Let the encoder and decoder functions be $f:\{0,1\}^{n}\rightarrow \{0,1\}^{\frac{n}{R}}$ and $g:\{0,1\}^{\frac{n}{R}}\rightarrow \{0,1\}^{n}$ respectively where $R$ is the rate of the code.
\end{itemize}
 For ease of analysis and tight upper bound for the number of tests we will use random codes and the optimal maximum-likelihood decoder which gives us the properties:
\begin{itemize}
\item There exists a sequence of codes $\{\mc{C}_n\}$ with the rate of each code being $R$ satisfying 
\begin{align}
\label{Eqn:ProbErrorCoding}
R<1-H(q)-\delta=1+q\log_2 q+ \overline{q}\log_2\overline{q}-\delta
\end{align}
for any arbitrary small constant $\delta$ such that the probability of error $\text{Pr}\left(g(\mbf{x+w})\neq \mbf{x}\right)<2^{-\kappa n}$ for some $\kappa >0$. In Eqn.~\ref{Eqn:ProbErrorCoding}, $\overline{q}\coleq 1-q$. 
\end{itemize} 
Even though the computational complexity of using random codes is exponential in block length of the code since the block length for our application is $\mc{O}(\log \frac{N}{K})$ and hence we have an overall computational complexity of $\mc{O}(N)$. But in practice one can use any of the popular error-correcting codes such as spatially-coupled LDPC codes or polar codes which are known to be capacity achieving \cite{kumar2014threshold,kudekar2013spatially} whose computational complexity is linear in block length. 

The modified signature matrix $\mbf{U}'_{r,p}$ can be described via $\mbf{U}_{r, p}$ given in Eq. \eqref{Eqn:SignatureMatrix} and encoding function $f$ for $\mc{C}_n$ where $n=\ceil{\log_2 r}$ as follows:
 \begin{align}
\label{Eqn:SignatureMatrixModified}
\mbf{U}_{r,p}'\coleq\begin{bmatrix}
f(\mbf{b}_1)  & f(\mbf{b}_2) &\cdots & f(\mbf{b}_r) \\
\overline{f(\mbf{b}_1)} & \overline{f(\mbf{b}_2)} & \cdots & \overline{f(\mbf{b}_r)}\\
f(\mbf{b}_{\pi^{1}_{1}}) & f(\mbf{b}_{\pi^{1}_{2}}) & \cdots & f(\mbf{b}_{\pi^{1}_{r}})\\
\overline{f(\mbf{b}_{\pi^{1}_{1}})} & \overline{f(\mbf{b}_{\pi^{1}_{2}})} & \cdots & \overline{f(\mbf{b}_{\pi^{1}_{r}})}\\
\cdots &  &\vdots \\
f(\mbf{b}_{\pi^{p-1}_{1}}) & f(\mbf{b}_{\pi^{p-1}_{2}}) & \cdots & f(\mbf{b}_{\pi^{p-1}_{r}})\\
\overline{f(\mbf{b}_{\pi^{p-1}_{1}})} & \overline{f(\mbf{b}_{\pi^{p-1}_{2}})} & \cdots & \overline{f(\mbf{b}_{\pi^{p-1}_{r}})}\\
\end{bmatrix}
\end{align}  
Then the overall testing matrix $\mbf{A}$ is defined in identical fashion to the definition in Sec. \ref{Sec:PriorWork} for the case of noiseless case except that $\mbf{U}$ will be replaced by $\mbf{U'}$ in Eqn.~\eqref{Eqn:SignatureMatrixModified}. Formally it can be defined as $\mbf{A}\coleq [\mbf{A}_{1}^{T},\ldots,\mbf{A}_{M_{1}}^{T}]^T$ where $\mbf{A}_{i}=\mbf{U'} \diag (\mbf{t}_i)$ where the binary vectors $\mbf{t}_i$ are defined in Sec. \ref{Sec:PriorWork}. 

\subsection*{Decoding}
The decoding scheme for the robust group testing, similar to the case of noiseless case, has two parts with the peeling part of the decoder identical to that of the noiseless case whereas the bin detection part differs slightly with an extra step of decoding for the error control code involved.

Given the test output vector at a bin $\mbf{y}=[\mbf{y}^{T}_{01},\mbf{y}^{T}_{02},\mbf{y}^{T}_{11},\ldots,\mbf{y}^{T}_{(p-1)2}]^T$, the bin detection for the noisy case can be summarized as following: The decoder $\forall i\in[0:p-1]$ applies the decoding function $g(\cdot)$ to the first segment $\mbf{y}_{i1}$ in each section $i$ and obtains the location $l_i$ whose binary expansion is equal to the error-correcting decoder output  $g(\mbf{y}_{i1})$. The decoder then declares the bin as a singleton if $\pi_{l_0}^{i}=l_i~\forall i$.

Similarly given that one of the variable nodes connected to the bin is already decoded to be non-zero, the resolvable double-ton decoding can be summarized as following. Let the location of the already recovered variable node in the bin (originally a double-ton) be $l_0$ then the test output can be given as 
\begin{align*}
\begin{bmatrix}
\mbf{y}_{01}\\
\mbf{y}_{02}\\
\mbf{y}_{11}\\
\vdots \\
\mbf{y}_{p2}\\
\end{bmatrix}
=\mbf{u}_{l_0} \vee \mbf{u}_{l_1}+\mbf{w}=
\begin{bmatrix}
f(\mbf{b}_{l_0})\\
\vspace{2pt}
\overline{f(\mbf{b}_{l_0})}\\
f(\mbf{b}_{\pi^1_{l_0}})\\
\vdots \\
\overline{f(\mbf{b}_{\pi^{p}_{l_0}})}\\
\end{bmatrix} 
\vee
\begin{bmatrix}
f(\mbf{b}_{l_1})\\
\vspace{2pt}
\overline{f(\mbf{b}_{l_1})}\\
f(\mbf{b}_{\pi^1_{l_1}})\\
\vdots \\
\overline{f(\mbf{b}_{\pi^{p}_{l_1}})}\\
\end{bmatrix}+
\begin{bmatrix}
\mbf{w}_{01}\\
\vspace{2pt}
\mbf{w}_{02}\\
\mbf{w}_{11}\\
\vdots \\
\mbf{w}_{p2}\\
\end{bmatrix}
\end{align*} 
where the location of the second non-zero variable node $l_1$ needs to be recovered. Given $\mbf{y}=\mbf{u}_{l_0} \vee \mbf{u}_{l_1}+\mbf{w}$ and $\mbf{u}_{l_0}$, the first segments of each section in $\mbf{u}_{l_1}+\mbf{w}$ can be recovered since for each segment of $\mbf{u}_{l_0}$ either the vector $f(\mbf{b}_{\pi^k_{l_0}})$ or it's complement is available. Once the first section $f(\mbf{b}_{\pi_{l_1}^i})+\mbf{w}$ of each segment $i$ is recovered, we apply singleton decoding procedure and rules as described above.

\begin{lemma}[Robust Bin Decoder Analysis]
\label{Lem:RobustBinAnalysis}
For a signature matrix $\mbf{U'}_{r,p}$ as described in \eqref{Eqn:SignatureMatrixModified}, the robust bin decoder misses a singleton with probability no greater than $\frac{p}{r^{\kappa}}$. The robust bin decoder wrongly declares a singleton with probability no greater than $\frac{1}{r^{p\kappa +p-1}}$.
\end{lemma}
\begin{proof}
Let $E_i$ be the event that the error-control decoder $g(\mbf{y}_{i1})$ commits an error at section $i$. From Eqn. \eqref{Eqn:ProbErrorCoding} we know that $\tx{Pr}(E_i)=2^{-\kappa \log r}=r^{-\kappa}$.
The robust bin decoder misses a singleton if the error-control decoder $g(\mbf{y}_{i1})$ commits an error at any one section. Thus the probability of missing a singleton can be upper bounded by applying union bound over all the sections $i\in[0:p-1]$ giving the required result. 

Consider a singleton bin and let the event where the robust bin decoder outputs a singleton hypothesis but the wrong index be $E_\tx{bin}$. This event happens when the error-control decoder commits an error and outputs the exact same wrong index at each and every section. We assume that given the error-control decoder makes an error, the output is uniformly random among all the remaining indices. Thus $\tx{Pr}(E_\tx{bin})$ can be upper bounded by $\frac{1}{r^\kappa} (\frac{1}{r^{1+\kappa}})^{p-1}$ which upon simplification gives us the required result.
\end{proof}

The fraction of missed singletons can be compensated by using $M(1+\frac{p}{r^{\kappa}})$ instead of $M$ such that the total number of singletons decoded will be $M(1+\frac{p}{r^{\kappa}})(1-\frac{p}{r^{\kappa}})\approx M$.
\begin{theorem}
Let $p\in\Z$ such that $K=o\left(N^{1-1/p}\right)$. The proposed robust regular SAFFRON scheme using $m=c\cdot K \log_{2}\frac{N\ell}{\ceps K}$ tests recovers atleast $(1-\epsilon)$ fraction of the defective items w.h.p. where $c=2p\beta(q)\ceps$ and $\beta(q)=1/R$. 
%
\end{theorem}
\begin{proof}
Similar to the noiseless case the total probability of error $P_e$ is dominated by the performance of bin decoder. 
\begin{align*}
P_e &\leq  \tx{Pr}(E_1)+ K\ell\tx{ Pr}(E_{\tx{bin}})\\
               &=\tx{Pr}(E_1)+ \mc{O}\left( \frac{K^{p+p\kappa}}{N^{p-1+p\kappa)}}\right)\\
   			   &= \mc{O}\left( \frac{N^{(p-1)(1+\kappa)}}{N^{p-1+p\kappa)}}\right)\\
               &\in \mc{O}(N^{-\kappa})
\end{align*}
where the second line is due to Lem. \ref{Lem:RobustBinAnalysis} and the third line is due to the fact that $\tx{Pr}(E_1)$ is exponentially decaying in $K$ and $K\leq N^{(p-1)/p}$ for large enough $K,N$.
\end{proof}

\section{Simulation Results}
In this section we will evaluate the performance of the proposed regular-SAFFRON scheme via Monte Carlo simulations and compare it with the results of SAFFRON scheme provided in \cite{lee2015saffron} for both the noiseless and noisy models.

\subsection*{Noiseless Group Testing}
As per Thm.~\ref{Thm:NoiselessMain} the proposed regular SAFFRON scheme requires only $6\ceps K\log \frac{N\ell}{\ceps K}$ tests as opposed to $6\ceps K\log N$ tests of SAFFRON scheme to recover $(1-\epsilon)$ fraction of defective items with a high probability. We demonstrate this by simulating the performance for the system parameters summarized below.
\begin{itemize}
\item We fix $N=2^{16}$ and $K=100$
\item For $\ell\in\{3,5,7\}$ we vary the number of bins $M=c K$. 
\item In Eqn.~\ref{Eqn:SignatureMatrix} the parameter $p=2$ is chosen for matrix $\mbf{U}$
\item Thus the bin detection size is $h=6\log_2 \frac{N\ell}{cK}$
\item Hence the total number of tests $m=6cK\log_2 \left(\frac{N\ell}{cK}\right)$
\end{itemize} 
The results are shown in Fig.~\ref{Fig:SimulationNoiseless}. We observe that there is clear improvement in performance for the proposed regular SAFFRON scheme when compared to the SAFFRON scheme for each $\ell\in\{3,5,7\}$.

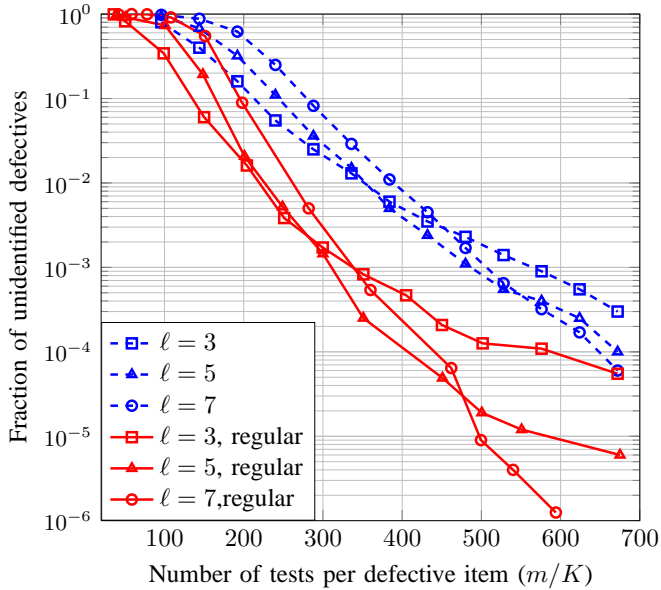
\begin{figure}[t!]
\centering
\resizebox{\columnwidth}{!}{
%
%
\begin{tikzpicture}
\def\fsize{\normalsize}
\pgfplotsset{every y tick label/.append style={font=\footnotesize}}
\pgfplotsset{every y tick label/.append style={font=\small}}

\begin{axis}[%
width=0.85\columnwidth,
height=0.8\columnwidth,
scale only axis,
xmin=20,
xmax=700,
xtick = {100,200,...,700},
xmajorgrids,
ymode=log,
ymin=1e-06,
ymax=1,
yminorticks=true,
ymajorgrids,
yminorgrids,
xlabel={\fsize{Number of tests per defective item ($m/K$)}},
ylabel={\fsize{Fraction of unidentified defectives}},
legend style={at={(0,0)},anchor=south west,draw=black,fill=white,legend cell align=left,font=\fsize}
]
\addplot [color=blue,dashed, mark=square,mark options={solid},line width=1pt]
  table[row sep=crcr]{96	0.8\\
144	0.4\\
192	0.16\\
240	0.055\\
288	0.025\\
336	0.013\\
384	0.006\\
432	0.0035\\
480	0.0023\\
528	0.0014\\
576	0.0009\\
624	0.00055\\
672	0.0003\\
};
\addlegendentry{$\ell=3$};

\addplot [color=blue,dashed,line width=1pt,mark=triangle,mark options={solid}]
  table[row sep=crcr]{96	0.94\\
144	0.68\\
192	0.32\\
240	0.11\\
288	0.036\\
336	0.015\\
384	0.005\\
432	0.0024\\
480	0.0011\\
528	0.00055\\
576	0.0004\\
624	0.00025\\
672	0.0001\\
};
\addlegendentry{$\ell=5$};

\addplot [color=blue,dashed,mark=o,mark options={solid},line width=1pt]
  table[row sep=crcr]{96	0.98\\
144	0.88\\
192	0.62\\
240	0.25\\
288	0.082\\
336	0.029\\
384	0.011\\
432	0.0045\\
480	0.0017\\
528	0.00065\\
576	0.00032\\
624	0.00017\\
672	6e-05\\
};
\addlegendentry{$\ell=7$};


\addplot [color=red,solid, line width=1pt, mark=square,mark options={solid}]
  table[row sep=crcr]{
36  0.99261\\
50.40      8.22e-1 \\
99.00      3.43e-1\\
150.00    6.01e-2 \\
204.00    1.61e-2\\
251.10    3.84e-3 \\
299.70    1.72e-3  \\
 351.00    8.28e-4 \\ 
 405.00    4.67e-4\\
 450.24    2.08e-4\\
 501.60  1.26e-4 \\  
 576.00  1.09e-4\\
672.00   5.50e-5\\
};
\addlegendentry{$\ell=3$, regular};

\addplot [color=red, solid, line width=1pt, mark=triangle,mark options={solid}]
  table[row sep=crcr]{
   39.000    0.970769\\
    50.700      9.16e-1   \\
100.800     7.33e-1   \\
148.500     1.93e-1   \\
201.000     2.07e-2  \\
249.000     5.22e-3  \\
300.000     1.450e-3\\
351.000     2.51e-4\\
450.900     4.90e-5\\
500.580     1.90e-5 \\
550.800     1.20e-5\\
675.000     6.00e-6\\
 };
\addlegendentry{$\ell=5$, regular};

\addplot [color=red, solid, line width=1.0pt,mark=o]
  table[row sep=crcr]{
42.00  0.99      \\
58.50	0.9995  \\
78.00    0.9937 \\
108.00   0.9142 \\
151.20  0.5487  \\
198.00 8.92e-2  \\
282.00  5.0e-3    \\
360.00  5.4e-4   \\
462.00  6.42e-5   \\
499.50  9.00e-6   \\
540.00  4.00e-6   \\
594.00  1.25e-6   \\
};
\addlegendentry{$\ell=7$,regular};

\end{axis}
\end{tikzpicture}%
}
\caption{MonteCarlo simulations for $K=100, N=2^{16}$. We compare the SAFFRON scheme \cite{lee2015saffron} with the proposed regular SAFFRON scheme for various left degrees $\ell\in\{3,5,7\}$. 
The plots in blue indicate the SAFFRON scheme and the plots in red indicate our regular SAFFRON scheme based on left-and-right-regular bipartite graphs.}
\label{Fig:SimulationNoiseless}
\end{figure}

\subsection*{Noisy Group Testing}
Similar to the noiseless group testing problem we simulate the performance of our robust regular-SAFFRON scheme and compare it with that of the SAFFRON scheme. For convenience of comparison we choose our system parameters identical to the choices in \cite{lee2015saffron}. The system parameters are summarized below:
\begin{itemize}
\item $N=2^{32}, K=2^7$. We fix $\ell=12, M=11.36K$ 
\item BSC noise parameter $q\in\{0.03,0.04,0.05\}$
\item In Eqn.~\ref{Eqn:SignatureMatrix} the parameter $p=1$ is chosen for matrix $\mbf{U}$
\item Thus the bin detection size is $h=4\log_2 \frac{N\ell}{M}$
\end{itemize}
The results are shown in Fig. \ref{Fig:SimulationNoisy}. Note that for the above set of parameters the right degree $r=\frac{N\ell}{M}\approx 26$. We choose to operate in field $GF(2^7)$ thus giving us a message length of $4$ symbols. For the choice of code we use a $(4+2e,4)$ Reed-Solomon code for $e\in[0:8]$ thus giving us a column length of $4\times 7(4+2e)$ bits at each bin and the total number of tests $m=28M(4+2e)$.

\begin{figure}[t!]
\centering
%
%
\begin{tikzpicture}
\def\fsize{\normalsize}
\pgfplotsset{every y tick label/.append style={font=\footnotesize}}
\pgfplotsset{every y tick label/.append style={font=\small}}

\begin{axis}[%
width=0.85\columnwidth,
height=0.8\columnwidth,
scale only axis,
xmin=1,
xmax=12,
xmajorgrids,
xtick = {4,6,8,10,12},
ymode=log,
ymin=1e-07,
ymax=1,
yminorticks=true,
ymajorgrids,
yminorgrids,
xlabel={\fsize{Number of tests ($\times 10^5$)}},
ylabel={\fsize{Fraction of unidentified defectives}},
legend style={at={(0,0)},anchor=south west,draw=black,fill=white,legend cell align=left, font=\tiny}
]
\addplot [color=blue,dashed,mark=square,mark options={solid}]
  table[row sep=crcr]{4.1877504	 0.12\\
5.5836672	0.011\\
6.979584 	0.0009\\
8.3755008	8e-05\\
9.7714176	1.9e-05\\
11.1673344	 8e-06\\
};
\addlegendentry{q=0.03};

\addplot [color=blue,dashed,mark=triangle,mark options={solid}]
  table[row sep=crcr]{4.1877504 	0.5\\
5.5836672		0.16\\
6.979584 		0.035\\
8.3755008		0.008\\
9.7714176	 	0.0016\\
11.1673344 	0.00035\\
};
\addlegendentry{q=0.04};

\addplot [color=blue,dashed,line width=1.0pt,mark=o,mark options={solid}]
  table[row sep=crcr]{4.1877504 	0.8\\
5.5836672		0.5\\
6.979584 		0.3\\
8.3755008		0.12\\
9.7714176		0.05\\
11.1673344 	0.02\\
};
\addlegendentry{q=0.05};

\addplot [color=red,solid,line width=1.0pt,mark=square,mark options={solid}]
  table[row sep=crcr]{
 1.63     0.1797 \\
 2.44     6.363e-3\\
3.26     2.948e-4\\
4.07     3.700e-5\\
4.89     1.16e-5\\
5.70    2.00e-6\\
};
\addlegendentry{q=0.03, regular};

\addplot [color=red,solid,line width=1.0pt,mark=triangle,mark options={solid}]
  table[row sep=crcr]{
2.44     5.469e-2\\
3.26     6.168e-3\\
4.89     1.250e-4\\
6.52     8.000e-6\\
};
\addlegendentry{q=0.04, regular};

\addplot [color=red,solid,line width=1.0pt,mark=o,mark options={solid}]
  table[row sep=crcr]{
   2.44  	1.302e-1  \\
   3.26  	3.348e-2\\
   4.89	    3.005e-3\\
   6.52	    4.340e-4\\
   8.15     1.5625e-5  \\
};
\addlegendentry{q=0.05, regular};

\end{axis}
\end{tikzpicture}
\caption{MonteCarlo simulations for $K=128, N=2^{32}$. We compare the SAFFRON scheme with the proposed regular-SAFFRON scheme for a left degree $\ell=12$. We fix the number of bins and vary the rate of the error control code used. The plots in blue indicate the SAFFRON scheme\cite{lee2015saffron} and the plots in red indicate the regular-SAFFRON scheme based on left-and-right-regular bipartite graphs.}
\label{Fig:SimulationNoisy}
\end{figure}
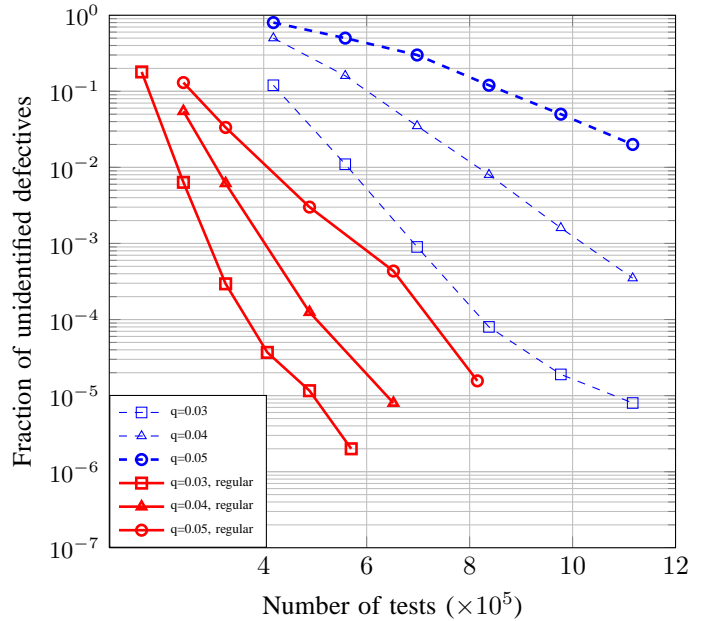

\section{Conclusion}
We addressed the Group Testing problem for $K$ defective items out of $N$ items. Using left-\emph{and}-right-regular sparse-graph codes we propose a new construction for the testing matrix based on that of Lee et al., \cite{lee2015saffron}. We show that this improves the testing complexity upon the previous results for the approximate version of the Group Testing problem and achieves asymptotically vanishing error probability under sub-linear time, order optimal, computational complexity. We also show that the proposed scheme with a variant of the original decoder has a testing complexity that is only $\log K$ factor away from the lower bound for the probabilistic version of the Group Testing problem with order optimal computational complexity.

\bibliographystyle{ieeetr}
\bibliography{journal_abbr,sparseestimation,grouptesting}
\end{document}